\RequirePackage{etoolbox}
\csdef{input@path}{{style/}{graphics/}}
\documentclass[aap,MSNbibl,nameyear,seceqn,dvips]{arximspdf}
\usepackage{url,breakurl}
\doi{10.1214/13-AAP997} 
\volume{25}
\issue{2}
\pubyear{2015}
\firstpage{465}
\lastpage{476}

\makeatletter
\newcommand{\rrVert}{\Vert}
\newcommand{\llVert}{\Vert}
\newtheorem{theorem}{Theorem}[section]
\newtheorem{lemma}[theorem]{Lemma}
\newproclaim{definition}[theorem]{Definition}
\newtheorem{fact}[theorem]{Fact}
\newproclaim{question}[theorem]{Question}
\newcommand{\pr}{\operatorname{Pr}}
\newcommand{\ignore}[1]{}
\newcommand{\poly}{\operatorname{poly}}
\newcommand{\hX}{\widehat{X}}
\makeatother

\begin{document}
\begin{frontmatter}

\title{A polynomial time approximation scheme for computing the supremum of Gaussian processes\thanksref{T1}}
\runtitle{Computing supremum of Gaussian processes}

\begin{aug}
\author[A]{\fnms{Raghu} \snm{Meka}\corref{}\ead[label=e1]{meka@microsoft.com}}
\runauthor{R. Meka}
\affiliation{Microsoft Research}
\address[A]{Microsoft Research\\
1288 Pear Avenue\\
Mountain View, California 94043\\
USA\\
\printead{e1}} 
\end{aug}
\thankstext{T1}{A preliminary version of this paper appeared in the
conference Foundations of Computer Science, 2012.}

\received{\smonth{7} \syear{2013}}

%
\begin{abstract}
We give a polynomial time approximation scheme (PTAS) for computing the
supremum of a Gaussian process. That is, given a finite set of vectors
$V \subseteq\mathbb{R}^d$, we compute a $(1+\varepsilon)$-factor
approximation to
$\mathop{\mathbb{E}}_{X \leftarrow\mathcal{N}^d}[\sup_{v \in V}|\langle v, X\rangle|]$ deterministically in
time $\poly(d)\cdot|V|^{O_\varepsilon(1)}$. Previously, only a constant
factor deterministic polynomial time approximation algorithm was known
due to the work of Ding, Lee and Peres
[\textit{Ann. of Math.} (\textit{2}) \textbf{175} (2012) 1409--1471].
This answers an open question of Lee (2010) and Ding [\textit{Ann. Probab.}
\textbf{42} (2014) 464--496].

The study of supremum of Gaussian processes is of considerable
importance in probability with applications in functional analysis,
convex geometry, and in light of the recent breakthrough work of Ding, Lee and Peres
[\textit{Ann. of Math.} (\textit{2}) \textbf{175} (2012) 1409--1471],
to random walks on finite graphs. As such our result could
be of use elsewhere. In particular, combining with the work of Ding
[\textit{Ann. Probab.}
\textbf{42} (2014) 464--496], our result yields a PTAS for computing the cover time of
bounded-degree graphs. Previously, such algorithms were known only for trees.

Along the way, we also give an explicit oblivious estimator for
semi-norms in Gaussian space with optimal query complexity. Our
algorithm and its analysis are elementary in nature, using two
classical \textit{comparison inequalities}, Slepian's lemma and Kanter's lemma.
\end{abstract}

%
\begin{keyword}[class=AMS]
\kwd[Primary ]{60C05}
\kwd[; secondary ]{68Q87}
\end{keyword}
\begin{keyword}
\kwd{Gaussian processes}
\kwd{derandomization}
\kwd{cover time}
\kwd{random walks}
\kwd{$\varepsilon$-nets}.
\end{keyword}

\end{frontmatter}

\setcounter{footnote}{1}
\section{Introduction}

The study of supremum of Gaussian processes is a major area of study in
probability and functional analysis as epitomized by the celebrated
\textit{majorizing measures} theorem of Fernique and Talagrand; see \citet{LedouxT}, \citet{Talagrand05} and references therein. There is by now a
rich body of work on obtaining tight estimates and characterizations of
the supremum of Gaussian processes with several applications in
analysis \citet{Talagrand05}, convex geometry \citet{Pisier99} and more.
Recently, in a striking result, \citet{DingLP11} used the theory to
resolve the \textit{blanket time} conjectures of \citet{WinklerZ96}.

\citet{DingLP11} used the powerful \textit{Dynkin isomorphism theory} and
majorizing measures theory to establish a structural connection between
the cover time (and blanket time) of a graph $G$ and the supremum of a
Gaussian process associated with the Gaussian Free Field on $G$. They
then use this connection to resolve the Winkler--Zuckerman blanket time
conjectures and to obtain the first deterministic polynomial time
constant factor approximation algorithm for computing the cover time of
graphs. This latter result resolves an old open question of \citet{AldousF}.


Besides showing the relevance of the study of Gaussian processes to
discrete combinatorial questions, the work of Ding, Lee and Peres gives
evidence that studying Gaussian processes could even be an important
algorithmic tool; a less investigated aspect in the rich literature on
Gaussian processes in probability and functional analysis. Here we
address the corresponding computational question directly, which given
the importance of Gaussian processes in probability, could be of use
elsewhere. In this context, the following question was asked by \citet
{leepost} and \citet{Ding11}.\hskip.2pt\footnote{We remark that \citet{leepost} and
\citet{Ding11} actually ask for an approximation to $\mathop{\mathbb{E}}_{X
\leftarrow\mathcal{N}
^d}[\sup_i \langle v _ i, X\rangle]$. However, this formulation
results in a
somewhat artificial asymmetry, and for most interesting cases these two\vspace*{1pt}
are essentially equivalent: if $\mathop{\mathbb{E}}_{X \leftarrow\mathcal
{N}^d}[\sup_i \langle v _ i, X\rangle]
= \omega(\max_i \llVert v_i\rrVert _2)$, then $\mathop{\mathbb{E}}_{X
\leftarrow\mathcal{N}^d}[\sup_I |\langle v _ i, X\rangle|] =
(1+o(1)) \mathop{\mathbb{E}}_{X \leftarrow\mathcal{N}^d}[\sup_i \langle v
_ i, X\rangle]$. We shall
overlook this distinction from now on.}

%
\begin{question}
For every $\varepsilon> 0$, is there a deterministic polynomial time
algorithm that, given a set of vectors $v_1,\ldots,v_m \in\mathbb{R}^d$,
computes a $(1 + \varepsilon)$-factor approximation to $\mathop{\mathbb{E}}_{X
\leftarrow\mathcal{N}
^d}[\sup_i |\langle v _ i,X\rangle|]$.\footnote{Throughout,
$\mathcal{N}$ denotes the
univariate Gaussian distribution with mean $0$ and variance $1$, and
for a distribution ${\mathcal D}$, $X \leftarrow{\mathcal D}$ denotes
a random
variable with distribution ${\mathcal D}$. By a $\alpha$-factor
approximation to a quantity $Z > 0$, we mean a number $p$ such that $p
\leq Z \leq\alpha p$.}
\end{question}

There is a simple randomized algorithm for the problem: sample a few
Gaussian vectors and output the median supremum value for the sampled
vectors. This, however, requires $O(d\log d/\varepsilon^2)$ random bits.
Using Talagrand's majorizing measures theorem, Ding, Lee and Peres give
a deterministic polynomial time $O(1)$-factor approximation algorithm
for the problem. This approach is inherently limited to not yield a
PTAS as the majorizing measures characterization is bound to lose a
universal constant factor. Here we give a PTAS for the problem thus
resolving the above question.

%
\begin{theorem}\label{thmain}
For every $\varepsilon> 0$, there is a deterministic algorithm that,
given a set of vectors $v_1,\ldots,v_m \in\mathbb{R}^d$, computes a
$(1 +
\varepsilon)$-factor approximation to $\mathop{\mathbb{E}}_{x \leftarrow
\mathcal{N}^d}[\sup_i |\langle v _ i,x\rangle|]$ in time $\poly(d)
\cdot m^{\widetilde{O}(1/\varepsilon^2)}$.
\end{theorem}

Our approach uses some classical \textit{comparison inequalities} in
convex geometry. To the best of our knowledge these inequalities have
not been used before in the context of algorithm design.

We explain our result on estimating semi-norms with respect to Gaussian
measures mentioned in the abstract in Section~\ref{seclinest}.

We next discuss some applications of our result to computing cover
times of graphs as implied by the works of \citet{DingLP11} and \citet{Ding11}.
\subsection{Application to computing cover times of graphs}
The study of random walks on graphs is an important area of research in
probability, algorithm design, statistical physics and more. As this is
not the main topic of our work, we avoid giving formal definitions and
refer the readers to \citet{AldousF}, \citet{Lovasz93} for background
information.

Given a graph $G$ on $n$-vertices, the cover time, $\tau_{\mathrm{cov}}(G)$, of
$G$ is defined as the expected time a random walk on $G$ takes to visit
all the vertices in $G$ when starting from the worst possible vertex in
$G$. Cover time is a fundamental parameter of graphs and is extensively
studied. Algorithmically, there is a simple randomized algorithm for
approximating the cover time---simulate a few trials of the random walk
on $G$ for $\poly(n)$ steps, and output the median cover time. However,
without randomness the problem becomes significantly harder. This was
one of the motivations of the work of \citet{DingLP11} who gave the
first deterministic constant factor approximation algorithm for the
problem, improving on an earlier work of \citet{KahnKLV00} who obtained
a deterministic $O((\log\log n)^2)$-factor approximation algorithm.
For the special case of trees, \citet{FeigeZ09} gave a FPTAS.

Ding, Lee and Peres also conjectured that the cover time of a graph $G$
(satisfying a certain reasonable technical condition) is asymptotically
equivalent to the supremum of an explicitly defined Gaussian process,
the Gaussian Free Field on~$G$. However, this conjecture though quite
interesting on its own, is not enough to give a PTAS for cover time;
one still needs a PTAS for computing the supremum of the relevant
Gaussian process. Our main result provides this missing piece, thus
removing one of the obstacles in their posited strategy to obtain a
PTAS for computing the cover time of graphs. Recently, \citet{Ding11}
showed the main conjecture of Ding, Lee and Peres to be true for
bounded-degree graphs and trees. Thus, combining his result [see
Theorem~1.1 in \citet{Ding11}] with Theorem~\ref{thmain}, we get a
PTAS for
computing cover time on bounded-degree graphs with $\tau_{\mathrm{hit}}(G) =
o(\tau_{\mathrm{cov}}(G))$.\footnote{The hitting time $\tau_{\mathrm{hit}}(G)$ is defined
as the maximum over all pairs of vertices $u,v \in G$ of the expected
time for a random walk starting at $u$ to reach $v$. See the discussion
in \citet{Ding11} for why this is a reasonable condition.} As mentioned
earlier, previously, such algorithms were only known for trees; see
\citet{FeigeZ09}.
\ignore{
%
\begin{theorem}
For every $\varepsilon> 0$ and $\Delta> 0$ there exists a constant
$C_{\Delta,\varepsilon}$ such that the following holds. For every graph
$G$ with maximum degree at most $\Delta$ and $\tau_{\mathrm{hit}}(G) <
C_{\Delta,\varepsilon}(\tau_{\mathrm{cov}}(G))$, there exists a deterministic
$n^{O_{\Delta,\varepsilon}(1)}$-time algorithm to compute a $(1+\varepsilon)$-factor
approximation to $\tau_{\mathrm{cov}}(G)$.
\end{theorem}
}

\section{Outline of algorithm}
The high level idea of our PTAS is as follows. Fix the set of vectors
$V = \{v_1,\ldots,v_m\} \subseteq\mathbb{R}^d$ and $\varepsilon>
0$. Without
loss of generality suppose that $\max_{v \in V} \llVert v\rrVert _2 =
1$. We first
reduce the dimension of $V$ by projecting $V$ onto a space of dimension
of $O((\log m)/\varepsilon^2)$ \'a la the classical Johnson--Lindenstrauss
lemma (JLL). We then give an algorithm that runs in time \mbox{polynomial} in
the number of vectors but exponential in the underlying dimension. Our
analysis relies on two comparison inequalities, Fernique--Slepian lemma [\citet
{Slepian62}] for the first step and Kanter's lemma [\citet{Kanter77}]
for the second step. We discuss these modular steps below.
\subsection{Dimension reduction} We project the set of vectors
$V\subseteq\mathbb{R}^d$ to $\mathbb{R}^k$ for $k = O((\log
m)/\varepsilon^2)$ to
preserve all pairwise (Euclidean) distances within a \mbox{$(1+\varepsilon
)$-}fac\-tor as in the Johnson--Lindenstrauss lemma (JLL). We then show
that the expected supremum of the \textit{projected} Gaussian process is
within a $(1 + \varepsilon)$ factor of the original value. The intuition
is that the supremum of a Gaussian process, though a global property,
can be controlled by pairwise correlations between the variables. To
quantify this, we use Slepian's lemma, that helps us relate the
supremum of two Gaussian processes by comparing pairwise correlations.
Finally, observe that using known derandomizations of JLL, the
dimension reduction can be done deterministically in time $\poly
(d,m,1/\varepsilon)$; see \citet{DJLL,Sivakumar02}.

Thus, to obtain a PTAS it would be enough to have a deterministic
algorithm to approximate the supremum of a Gaussian process in time
exponential in the dimension $k = O((\log m)/\varepsilon^2)$.
Unfortunately, a naive argument by discretizing the Gaussian measure in
$\mathbb{R}^k$ leads to a run-time of at least $k^{O(k)}$; which gives a
$m^{O((\log\log m)/\varepsilon^2)}$ algorithm. This question was recently
addressed by \citet{DadushV12}, who needed a similar sub-routine for
their work on computing \textit{M-Ellipsoids} of convex sets and give a
deterministic algorithm with a run-time of $(\log k)^{O(k)}$. We
resolve this question fully by giving an optimal oblivious estimator
for norms in Gaussian space, which when combined with the dimension
reduction step gives a PTAS for computing the supremum. 
\subsection{Oblivious estimators for semi-norms}\label{seclinest}
Let $\varphi\dvtx \mathbb{R}^k \rightarrow\mathbb{R}_+$ be a semi-norm,
that is, $\varphi$ is
homogeneous and satisfies triangle inequality. For normalization
purposes, we assume that $1 \leq\mathop{\mathbb{E}}_{x \leftarrow\mathcal
{N}^k}[\varphi(x)]$ and that
the Lipschitz constant of $\varphi$ is at most $k^{O(1)}$. This is
satisfied in most reasonable cases. Note that the supremum function
$\varphi_V(x) = \sup_{v \in V}|\langle v,  x\rangle|$ satisfies
these conditions. Our
goal will be to compute a $(1+\varepsilon)$-factor approximation to
$\mathbb{E}
_{x \leftarrow\mathcal{N}^k}[\varphi(x)]$ in time $2^{O_\varepsilon(k)}$.
\ignore{
%
\begin{theorem}\label{thepsnet}
For every $\varepsilon> 0$, there exists a distribution ${\mathcal D}$ on
$\mathbb{R}^k$ which can be sampled using $O(k\log(1/\varepsilon))$
bits in time
$\poly(k,1/\varepsilon)$ and space $O(\log k + \log(1/\varepsilon
))$ such
that for every semi-norm $\varphi\dvtx \mathbb{R}^k \rightarrow\mathbb{R}_+$,
\[
(1-\varepsilon) \mathop{\mathbb{E}}_{x \leftarrow{\mathcal D}}\bigl[\varphi(x)\bigr] \leq
\mathop{\mathbb{E}}_{x \leftarrow\mathcal{N}
^k}\bigl[\varphi(x)\bigr] \leq(1+\varepsilon)
\mathop{\mathbb{E}}_{x \leftarrow
{\mathcal D}}\bigl[\varphi(x)\bigr].
\]
In particular, there exists a deterministic $(1/\varepsilon)^{O(k)}$-time
algorithm for computing a $(1+\varepsilon)$-factor approximation to
$\mathbb{E}
_{X \leftarrow\mathcal{N}^k}[\varphi(X)]$ using only oracle access
to $\varphi$.
\end{theorem}
}

%
\begin{theorem}\label{thepsnetintro}
For every $\varepsilon> 0$, there exists a deterministic algorithm
running in time $(1/\varepsilon)^{O(k)}$ and space $\poly
(k,1/\varepsilon)$
that computes a $(1+\varepsilon)$-factor approximation to $\mathbb
{E}_{X \leftarrow\mathcal{N}
^k}[\varphi(X)]$ using only oracle access to $\varphi$.
\end{theorem}

Our algorithm has the additional property of being an \textit{oblivious
linear estimator}: the set of query points does not depend on $\varphi$,
and the output is a positive weighted sum of the evaluations of
$\varphi$
on the query points. Further, the construction is essentially optimal
as any such oblivious estimator needs to make at least $(1/\varepsilon
)^{\Omega(k)}$ queries; see Section~\ref{secappendix}. In
comparison, the
previous best bound of Dadush and Vempala [\citet{DadushV12}] needed
$((\log k)/\varepsilon)^{O(k)}$ queries. 

A natural first approach to compute $\mathop{\mathbb{E}}_{X \leftarrow\mathcal
{N}^k}[\varphi(X)]$,
would be to first discretize the one-dimensional Gaussian distribution
with a constant granularity $\delta= f(\varepsilon)$ to get a
distribution $\mu$ and then evaluate the expectation with respect to
the product distribution $\mu^k$. We will show that this seemingly
naive approach in fact does very well, giving an error bound that does
not depend on the dimension $k$. We do so by using a classical
comparison inequality---Kanter's lemma---that allows us to ``lift'' a
simple estimator for the univariate case to the multi-dimensional case.

More concretely, we first construct a symmetric distribution $\mu$ on
$\mathbb{R}$ that has a simple \textit{piecewise flat graph} and \textit{sandwiches}
the one-dimensional Gaussian distribution in the following sense. Let
$\nu$ be a ``shrinking'' of $\mu$ defined to be the \mbox{probability} density
function (p.d.f.) of $(1-\varepsilon)x$ for $x \leftarrow\mu$. We show
that if
$\mu$ has \textit{granularity} about $\varepsilon^{3/2}$, then, for every
symmetric interval $I \subseteq\mathbb{R}$, $\mu(I) \leq\mathcal
{N}(I) \leq\nu(I)$.

Kanter's lemma [\citet{Kanter77}] then says that for p.d.f.'s $\mu,\nu$ as
above that are in addition \textit{unimodal}, the above relation carries
over to the product distributions $\mu^k, \nu^k$: for every symmetric
convex set $K \subseteq\mathbb{R}^k$, $\mu^k(K) \leq\mathcal
{N}^k(K) \leq\nu^k(K)$.
This last inequality immediately implies that semi-norms cannot \textit{distinguish} between $\mu^k$ and $\mathcal{N}^k$: for any semi-norm
$\varphi$, $\mathbb{E}
_{\mu^k}[\varphi(x)] = (1\pm\varepsilon)\mathop{\mathbb{E}}_{\mathcal
{N}^k}[\varphi(x)]$. We then
suitably prune the distribution $\mu^k$ to have small support and prove
Theorem~\ref{thepsnet}.

Our main result, Theorem~\ref{thmain}, follows by first reducing the
dimension as in the previous section and applying Theorem~\ref
{thepsnet} to
the semi-norm $\varphi\dvtx \mathbb{R}^k \rightarrow\mathbb{R}_+$,
$\varphi(x) = \sup_i|\langle u _ i,x\rangle|$
for the projected vectors $\{u_1,\ldots,u_m\}$.

\section{Dimension reduction}
The use of JLL type random projections for estimating the supremum
comes from the following comparison inequality for Gaussian processes.
We call a collection of real-valued random variables $\{X_t\}_{t \in
T}$ a~Gaussian process if every finite linear combination of the
variables has a normal distribution with mean zero. We refer the reader
to Corollary~3.14 and the following discussion in \citet{LedouxT} for
reference. 

%
\begin{theorem}[(Fernique--Slepian lemma)]\label{lmslepian}
  Let $\{X_t\}_{t \in T}$ and $\{Y_t\}_{t \in T}$ be two
  Gaussian processes such that for every $s,t \in T$,
  $\mathbb{E}[(X_s - X_t)^2] \leq \mathbb{E}[(Y_s - Y_t)^2]$. Then,
  $\mathbb{E}[\sup_t |X_t|] \leq \mathbb{E}[\sup_t |Y_t|]$.
\end{theorem}

%
We also need a derandomized version of the Johnson--Lindenstrauss lemma.

\begin{theorem}[{[\citet{DJLL}]}]\label{thdjll}
For every $\varepsilon> 0$, there exists a deterministic $(d m^2 (\log m
+ 1/\varepsilon)^{O(1)})$-time algorithm that given vectors
$v_1,\ldots,v_m \in\mathbb{R}^d$ computes a linear mapping $A\dvtx \mathbb{R}^d
\rightarrow\mathbb{R}^k$ for $k =
O((\log m)/\varepsilon^2)$ such that for every $i,j \in[m]$, $\llVert
v_i - v_j\rrVert _2 \leq\llVert A(v_i) - A(v_j)\rrVert _2 \leq
(1+\varepsilon)\llVert v_i - v_j\rrVert _2$.
\end{theorem}

Combining the above two theorems immediately implies the following.

\begin{lemma}\label{lmderandjl}
For every $\varepsilon> 0$, there exists a deterministic $(d m^2 (\log m
+ 1/\varepsilon)^{O(1)})$-time algorithm that given vectors
$v_1,\ldots,v_m \in\mathbb{R}^d$ computes a linear mapping $A\dvtx \mathbb{R}^d
\rightarrow\mathbb{R}^k$ for $k =
O((\log m)/\varepsilon^2)$ such that
%
%
\begin{eqnarray}\label{eqjllsup}
\mathop{\mathbb{E}}_{x \leftarrow\mathcal{N}^d}\Bigl[\sup_i \bigl|
\langle v _ i,x\rangle\bigr|\Bigr] &\leq& \mathop{\mathbb{E}}_{y \leftarrow\mathcal{N}^k}\Bigl[\sup
_i \bigl|\bigl\langle A, ( v_i),y\bigr\rangle\bigr|\Bigr]
\nonumber\\[-8pt]\\[-8pt]
&\leq& (1+\varepsilon) \mathop{\mathbb{E}}_{x
\leftarrow\mathcal{N}^d}\Bigl[\sup_i \bigl|
\langle v _ i,x\rangle\bigr|\Bigr].\nonumber
\end{eqnarray}
\end{lemma}

\begin{pf}
Let $V = \{v_1,\ldots,v_m\} \cup\{-v_1,\ldots,-v_m\}$, and let $\{
X_v\}_{v \in V}$ be the Gaussian process where the joint distribution
is given by $X_v \equiv\langle v,  x\rangle$ for $x \leftarrow
\mathcal{N}^d$. Then $\mathop{\mathbb{E}}_{x
\leftarrow\mathcal{N}^d}[\sup_i |\langle v _ i,x\rangle|] = \mathbb
{E}[\sup_v X_v]$.\vspace*{1pt}

Let $A\dvtx \mathbb{R}^d \rightarrow\mathbb{R}^k$ be the linear mapping
as given by Theorem~\ref{thdjll} applied to $V$. Let $\{Y_v\}_{v \in V}$ be the ``projected''
Gaussian process with joint distribution given by $Y_v \equiv  \langle
A, ( v),y\rangle$ for $y \leftarrow\mathcal{N}^k$. Then $\mathbb
{E}_{y \leftarrow\mathcal{N}^k}[\sup_i |\langle v _ i,y\rangle|] =
\mathbb{E}[\sup_v Y_v]$.

Finally, observe that for any $u,v \in V$,
\begin{eqnarray*}
\mathbb{E}\bigl[(X_u - X_v)^2\bigr] &=&
\llVert u-v\rrVert _2^2 \leq\bigl\llVert A(u) - A(v)\bigr
\rrVert _2^2
\\
&=& \mathbb{E}\bigl[(Y_u -
Y_v)^2\bigr] \leq(1+\varepsilon)^2\mathbb{E}
\bigl[(X_u - X_v)^2\bigr].
\end{eqnarray*}

Combining the above inequality with Lemma~\ref{lmslepian}
applied to the pairs of processes $ (\{X_v\}_{v \in V}, \{Y_v\}_{v
\in V} )$ and $ (\{Y_v\}_{v \in V}, \{(1+\varepsilon)X_v\}
_{v \in
V} )$ it follows that
 \[
 \mathbb{E}\Bigl[\sup_v |X_v|\Bigr] \leq \mathbb{E}\Bigl[\sup_v |Y_v|\Bigr] \leq
 \mathbb{E}\Bigl[\sup_v (1+\varepsilon)|X_v|\Bigr] = (1+\varepsilon)
 \mathbb{E}\Bigl[\sup_v |X_v|\Bigr].
 \]
%
The lemma now follows.
\end{pf}

\section{Oblivious estimators for semi-norms in Gaussian space}\label
{secepsnet}
In the previous section we reduced the problem of computing the
supremum of a $d$-dimensional Gaussian process to that of a Gaussian
process in $k = O((\log m)/\break \varepsilon^2)$-dimensions. Thus it
suffices to
have an algorithm for approximating the supremum of Gaussian processes
in time exponential in the dimension. We will give such an algorithm
that works more generally for all semi-norms.

Let $\varphi\dvtx \mathbb{R}^k \rightarrow\mathbb{R}_+$ be a semi-norm.
That is, $\varphi$ satisfies
the triangle inequality and is homogeneous. For normalization purposes
we assume that $1 \leq\mathop{\mathbb{E}}_{\mathcal{N}^k}[\varphi(X)]$ and
the Lipschitz constant
of $\varphi$ is at most $k^{O(1)}$. 

%
\begin{theorem}\label{thepsnet}
For\vspace*{1pt} every $\varepsilon> 0$, there exists a set $S \subseteq\mathbb
{R}^k$ with
$|S| = (1/\varepsilon)^{O(k)}$ and a function $p\dvtx \mathbb{R}^k
\rightarrow\mathbb{R}_+$
computable in $\poly(k,1/\varepsilon)$ time such that the following holds.
For every semi-norm $\varphi\dvtx \mathbb{R}^k \rightarrow\mathbb{R}_+$,
\[
(1 - \varepsilon) \biggl(\sum_{x \in S} p(x) \varphi(x)
\biggr) \leq \mathop{\mathbb{E}}_{X
\leftarrow\mathcal{N}^k}\bigl[\varphi(X)\bigr] \leq(1 + \varepsilon)
\biggl(\sum_{x \in S} p(x) \varphi (x) \biggr).
\]
Moreover, successive elements of $S$ can be enumerated in $\poly
(k,1/\varepsilon)$ time and $O(k\log(1/\varepsilon))$ space.
\end{theorem}

Theorem~\ref{thepsnetintro} follows immediately from the above.

\begin{pf*}{Proof of Theorem~\ref{thepsnetintro}}
Follows by enumerating over the set $S$ and computing $\sum_{x \in S}
p(x) \varphi(x)$ by querying $\varphi$ on the points in $S$. 
\end{pf*}

We now prove Theorem~\ref{thepsnet}.
Here and henceforth, let $\gamma$ denote the p.d.f. of the standard
univariate Gaussian distribution. Fix $\varepsilon> 0$, and let
$\delta>
0$ be a parameter to be chosen later. Let $\mu\equiv\mu_{\delta}$ be
the p.d.f. which is a piecewise-flat approximator to $\gamma$ obtained by
spreading the mass $\gamma$ gives to an interval $I = [i\delta,
(i+1)\delta)$ evenly over $I$. Formally, $\mu(z) = \mu(-z)$ and for $z
> 0$, $z \in[i\delta,(i+1)\delta)$,
%
%
\begin{equation}
\label{eqdefmu} \mu(z) = \frac{\gamma([i\delta,(i+1)\delta))}{\delta}.
\end{equation}
Clearly, $\mu$ defines a symmetric distribution on $\mathbb{R}$. We
will show
that for $\delta\ll\varepsilon$ sufficiently small, semi-norms cannot
\textit{distinguish} the product distribution $\mu^k$ from~$\mathcal{N}^k$:

\begin{lemma}\label{lmepsnetm}
Let $\delta= (2\varepsilon)^{3/2}$. Then, for every semi-norm
$\varphi\dvtx \mathbb{R}^k
\rightarrow\mathbb{R}$,
\[
(1-\varepsilon) \mathop{\mathbb{E}}_{X \leftarrow\mu^k}\bigl[\varphi(X)\bigr] \leq
\mathop{\mathbb{E}}_{Z \leftarrow\mathcal{N}^k}\bigl[\varphi (Z)\bigr] \leq\mathop{\mathbb{E}}_{X \leftarrow\mu^k}
\bigl[\varphi(X)\bigr].
\]
\end{lemma}

We first prove Theorem~\ref{thepsnet} assuming the above lemma, whose proof
is deferred to the next section.

\begin{pf*}{Proof of Theorem~\ref{thepsnet}}
Let $\hat{\mu}$ be the symmetric distribution supported on $\delta
(\mathbb{Z}+
1/2)$ with p.d.f. defined by
\[
\hat{\mu}\bigl(\delta(i+1/2)\bigr) = \mu\bigl(\bigl[i\delta, (i+1)\delta\bigr)\bigr)
\]
for $i \geq0$. Further, let $X \leftarrow\mu^k$, $\hX\leftarrow
\hat{\mu}^k$,
$Z \leftarrow\mathcal{N}^k$.

We claim that $\mathbb{E}[\varphi(\hX)] = (1\pm\varepsilon)\mathbb
{E}[\varphi(Z)]$. Let $Y$
be uniformly distributed on $[-\delta,\delta]^k$ and observe that
random variable $X \equiv\hX+ Y$ in law. Therefore,
%
\begin{eqnarray}
\mathbb{E}\bigl[\varphi(X)\bigr] &=& \mathbb{E}\bigl[\varphi(\hX+Y)\bigr] =
\mathbb {E}\bigl[\varphi(\hX)\bigr] \pm\mathbb{E}\bigl[\varphi(Y)\bigr] \nonumber
\\
&=&
\mathbb{E} \bigl[\varphi(\hX)\bigr] \pm\delta\mathbb{E}\bigl[\varphi(Y/\delta)
\bigr]
\nonumber\\[-8pt]\\[-8pt]
&=& \mathbb{E}\bigl[\varphi(\hX)\bigr] \pm\delta\mathop{\mathbb{E}}_{Z' \in_u
[-1,1]^k}\bigl[
\varphi\bigl(Z'\bigr)\bigr] \nonumber
\\
&=& \mathbb{E}\bigl[\varphi (\hX)\bigr] \pm
\delta\mathbb{E}\bigl[\varphi(Z)\bigr]\qquad\mbox{(Lemma \ref
{lmcube})}.\nonumber
\end{eqnarray}
Thus, by Lemma \ref{lmepsnetm},
%
%
\begin{equation}
\label{eqlm1} \mathbb{E}\bigl[\varphi(\hX)\bigr] = \bigl(1\pm O(\varepsilon)
\bigr) \mathbb {E}\bigl[\varphi(Z)\bigr].
\end{equation}
We next prune $\hat{\mu}^k$ to reduce its support. Define $p\dvtx \mathbb
{R}^k \rightarrow
\mathbb{R}_+$ by $p(x) = \hat{\mu}^k(x)$. Clearly, $p(x)$ being a product
distribution is computable in $\poly(k,1/\varepsilon)$ time.

Let $S =  (\delta(\mathbb{Z}+ 1/2) )^k \cap B_2(3\sqrt {k})$, where
$B_2(r) \subseteq\mathbb{R}^k$ denotes the Euclidean ball of radius
$r$. As
$\varphi$ has Lipschitz\vspace*{1pt} constant bounded by $k^{O(1)}$, a simple
calculation shows that throwing away all\vspace*{1pt} points in the support of $\hX$
outside $S$ does not change $\mathbb{E}[\varphi(\hX)]$ much. It is
easy to check
that for $x \notin S$, $p(x) \leq\exp(-\llVert x\rrVert
_2^2/4)/(2\pi)^{k/2}$. Therefore,
%
\begin{eqnarray}
\mathbb{E}\bigl[\varphi(\hX)\bigr] &=& \sum_{x} p(x)
\varphi(x) = \sum_{x \in
S}p(x) \varphi(x) + \sum
_{x \notin S}p(x) \varphi(x)\nonumber
\\
&=& \sum_{x \in S} p(x) \varphi(x) \pm\sum
_{x \notin S} \frac{\exp
(-\llVert x\rrVert _2^2/4)}{(2\pi)^{k/2}}\cdot\bigl( k^{O(1)} \llVert x
\rrVert _2\bigr)
\\
&=& \sum_{x \in S}p(x)\varphi(x) \pm o(1).\nonumber
\end{eqnarray}

From equation (\ref{eqlm1}) and the above equation, we get (recall
that $\mathbb{E}[\varphi
(Z)] \geq1$)
\[
\mathbb{E}\bigl[\varphi(Z)\bigr] = \bigl(1\pm O(\varepsilon)\bigr) \biggl(\sum
_{x \in
S} p(x) \varphi (x) \biggr),
\]
which is what we want to show.

We now reason about the complexity of $S$. First, by a simple covering
argument $|S| < (1/\delta)^{O(k)}$,
\[
|S| < \frac{\mathrm{Vol} (B_2(3\sqrt{k}) + [-\delta,\delta]^k)}{\mathrm{Vol}
([-\delta,\delta]^k)} = (1/\delta)^{O(k)} = (1/\varepsilon)^{O(k)},
\]
where for sets $A, B \subseteq\mathbb{R}^k$, $A+B$ denotes the
Minkowski sum,
and $\mathrm{Vol}$ denotes Lebesgue volume. This size bound almost suffices to
prove Theorem~\ref{thepsnet} except for the complexity of enumerating
elements from $S$. Without loss of generality assume that $R = 3\sqrt {n}/\delta$ is an integer. Then, enumerating elements in $S$ is
equivalent to enumerating integer points in the $n$-dimensional ball of
radius $R$. This can be accomplished by going through the set of
lattice points in the natural lexicographic order, and takes $\poly
(k,1/\varepsilon)$ time and $O(k\log(1/\varepsilon))$ space per
point in $S$.
\end{pf*}

%
\section{Proof of Lemma \texorpdfstring{\protect\ref{lmepsnetm}}{4.2}}
Our starting point is the following definition that helps us \textit{compare} multivariate distributions when we are only interested in
volumes of convex sets. We shall follow the notation of \citet{Ball}.

%
\begin{definition}
Given\vspace*{2pt} two symmetric p.d.f.'s, $f,g$ on $\mathbb{R}^k$, we say that $f$ is less
peaked than $g$ ($f \preceq g$) if for every symmetric convex set $K
\subseteq\mathbb{R}^k$, $f(K) \leq g(K)$.
\end{definition}

We also need the following elementary facts. The first follows from the
unimodality of the Gaussian density and the second from partial integration.

%
\begin{fact}
For any $\delta> 0$ and $\mu$ as defined by equation (\ref
{eqdefmu}), $\mu$ is
less peaked than $\gamma$.
\end{fact}

%
\begin{fact}\label{fctpeaked}
Let $f, g$ be distributions on $\mathbb{R}^k$ with $f \preceq g$. Then
for any
semi-norm $\varphi\dvtx \mathbb{R}^k \rightarrow\mathbb{R}$, $\mathbb
{E}_f[\varphi(x)] \geq\mathbb{E}_g[\varphi(x)]$.
\end{fact}

\begin{pf}
Observe that for any $t > 0$, $\{x\dvtx  \varphi(x) \leq t\}$ is convex. Let
random \mbox{variables} $X \leftarrow f$, $Y \leftarrow g$. Then, by partial
integration, $\mathbb{E}[\varphi(X)] = \int_0^\infty\varphi'(t)
\* \pr[ \varphi(X) > t]
\,dt \geq\int_0^\infty\varphi'(t) \pr[\varphi(Y)> t] \,dt = \mathbb
{E}[\varphi(Y)]$.
\end{pf}

The above statements give us a way to compare the expectations of $\mu$
and $\gamma$ for one-dimensional convex functions. We would now like to
do a similar comparison for the product distributions $\mu^k$ and
$\gamma^k$. For this we use Kanter's lemma [\citet{Kanter77}], which says
that the relation $\preceq$ is preserved under tensoring if the
individual distributions have the additional property of being \textit{unimodal}.

\begin{definition}
A distribution $f$ on $\mathbb{R}^n$ is unimodal if $f$ can be written
as an
increasing limit of a sequence of distributions each of which is a
finite positively weighted sum of uniform distributions on symmetric
convex sets.
\end{definition}


%
\begin{theorem}[(Kanter's lemma [\citet{Kanter77}{]}; cf.~\citet{Ball})]\label{thkanter}
Let $\mu_1,\mu_2$ be symmetric distributions on $\mathbb{R}^n$ with
$\mu_1
\preceq\mu_2$ and let $\nu$ be a unimodal distribution on $\mathbb{R}^m$.
Then, the product distributions $\mu_1 \times\nu$, $\mu_2 \times
\nu$
on $\mathbb{R}^n \times\mathbb{R}^m$ satisfy $\mu_1 \times\nu
\preceq\mu_2 \times\nu$.
\end{theorem}

We next show that $\mu$ ``sandwiches'' $\gamma$ in the following sense.

\begin{lemma}
Let $\nu$ be the p.d.f. of the random variable $y = (1-\varepsilon)x$
for \mbox{$x
\leftarrow\mu$}. Then, for $\delta\leq(2\varepsilon)^{3/2}$, $\mu
\preceq
\gamma\preceq\nu$.
\end{lemma}

\begin{pf}
As mentioned above, $\mu\preceq\gamma$. We next show that $\gamma
\preceq\nu$. Intuitively, $\nu$ is obtained by spreading the mass that
$\gamma$ puts on an interval $I = [i\delta, (i+1)\delta)$ evenly on the
\emph{smaller} interval $(1-\varepsilon)I$. The net effect of this
operation is to push the p.d.f. of $\mu$ closer to the origin and for
$\delta$ sufficiently small the inward push from this ``shrinking''
wins over the outward push from going to $\mu$.

Fix an interval $I = [-i \delta(1-\varepsilon) - \theta, i\delta
(1-\varepsilon
) + \theta]$ for $0 \leq\theta< \delta(1-\varepsilon)$. Then
%
%
\begin{eqnarray}
\label{eqcases} \nu(I) &=& \nu \bigl( \bigl[-i\delta(1-\varepsilon), i\delta (1-
\varepsilon)\bigr] \bigr) + 2 \nu \bigl( \bigl[i\delta(1-\varepsilon), i\delta (1-
\varepsilon) + \theta\bigr] \bigr)
\\
&=& \gamma \bigl( [-i\delta,i\delta] \bigr) + \frac{2 \theta
\cdot
\gamma( [i\delta,(i+1)\delta) )}{\delta(1-\varepsilon)}.
\end{eqnarray}
We now consider two cases.
\begin{longlist}[\textit{Case} 1:]
\item[\textit{Case} 1:] $i \geq(1-\varepsilon)/\varepsilon$ so that $i\delta
(1-\varepsilon) +
\theta\leq i\delta$. Then, from the above equation,
\[
\nu(I) \geq\gamma \bigl( [-i\delta, i \delta] \bigr) \geq \gamma \bigl( \bigl[-i
\delta(1-\varepsilon)-\theta, i\delta(1-\varepsilon)+\theta \bigr] \bigr) =
\gamma(I).
\]

\item[\textit{Case} 2:] $i < (1-\varepsilon)/\varepsilon$. Let $\alpha= (i+1)\delta
= \delta
/\varepsilon$. Then, as $1 - x^2/2 \leq e^{-x^2/2} \leq1$,
\[
\gamma\bigl((i\delta, i\delta+ \theta]\bigr) \leq\theta\cdot\gamma(0), \gamma\bigl(
\bigl[i\delta,(i+1)\delta\bigr) \bigr) \geq\delta\cdot\gamma(0) \cdot \bigl(1-
\alpha^2/2\bigr).
\]
\end{longlist}

Therefore,
\begin{eqnarray*}
\nu(I) &=& \gamma(I) - 2 \gamma \bigl( \bigl(i\delta, i\delta (1-\varepsilon) +
\theta\bigl] \bigr) + \frac{2 \theta\cdot\gamma( [i\delta,(i+1)\delta)
)}{\delta(1-\varepsilon)}
\\
&\geq&\gamma(I) - 2 \gamma \bigl( (i\delta, i\delta+ \theta] \bigr) + \frac{2 \theta\cdot\gamma( [i\delta,(i+1)\delta) )}{\delta
(1-\varepsilon)}
\\
&\geq&\gamma(I) - 2 \theta\gamma(0) + \frac{2 \theta\cdot\delta
\cdot\gamma(0) \cdot(1-\alpha^2/2)}{\delta(1-\varepsilon)}
\\
&=& \gamma(I) + \frac{2\theta\gamma(0)}{1-\varepsilon} \cdot \bigl(\varepsilon- \alpha^2/2
\bigr) \geq\gamma(I),
\end{eqnarray*}
for $\alpha^2 \leq2 \varepsilon$, that is, if $\delta\leq
(2\varepsilon)^{3/2}$.
\ignore{
Let $\delta\ll\varepsilon$ be sufficiently small so that the p.d.f. of
$\gamma$ is nearly constant in the interval $[0,(i+1)\delta]$, that is
$\gamma([\alpha,\beta]) = (\beta-\alpha) \gamma(0) \pm O((\beta
-\alpha
)^2)$ for $0 < \alpha< \beta< (i+1) \delta$. Then, by equation (\ref
{eqcases}), as $\theta< \delta(1-\varepsilon)$,
\begin{eqnarray*}
\nu(I) &\geq&\gamma(I) - 2 \gamma ( \bigl(i\delta, i\delta (1-\varepsilon) +
\theta] \bigr) + \frac{2 \theta\cdot\gamma( [i\delta,(i+1)\delta
) )}{\delta(1-\varepsilon)}
\\
&\geq&\gamma(I) - 2 \gamma ( (i\delta, i\delta+ \theta] ) + \frac{2 \theta\cdot\gamma( [i\delta,(i+1)\delta) )}{\delta
(1-\varepsilon)}
\\
&\geq&\gamma(I) - 2 \theta\gamma(0) - O\bigl(\theta^2\bigr) +
\frac{2 \theta
(\delta\gamma(0) - O(\delta^2))}{\delta(1-\varepsilon)}
\\
&=& \gamma(I) + \frac{2\theta\gamma(0)\varepsilon}{1-\varepsilon} - O(\theta \delta) \geq\gamma(I)
\end{eqnarray*}
}
\end{pf}

Lemma \ref{lmepsnetm} follows easily from the above two claims.

\begin{pf*}{Proof of Lemma \ref{lmepsnetm}}
Clearly, $\mu,\nu,\gamma$ are unimodal and product of unimodal
distributions is unimodal. Thus, from the above lemma and iteratively
\mbox{applying} Kanter's lemma we get $\mu^k \preceq\gamma^k \preceq\nu^k$.
Therefore, by Fact \ref{fctpeaked}, for any semi-norm $\varphi$,
\[
\mathop{\mathbb{E}}_{\mu^k}\bigl[\varphi(X)\bigr] \geq\mathop{\mathbb{E}}_{\gamma^k}
\bigl[\varphi (Y)\bigr] \geq\mathop{\mathbb{E}}_{\nu^k}\bigl[\varphi (X)\bigr] =
\mathop{\mathbb{E}}_{\mu^k}\bigl[\varphi\bigl((1-\varepsilon)X\bigr)\bigr] = (1-
\varepsilon)\mathop{\mathbb{E}}_{\mu^k}\bigl[\varphi(X)\bigr].
\]\upqed
\end{pf*}

We now prove the auxiliary lemma we used in proof of Theorem~\ref{thepsnet}.

\begin{lemma}\label{lmcube}
Let $\rho$ be the uniform distribution on $[-1,1]$. Then, $\gamma
\preceq\rho$ and for any semi-norm $\varphi\dvtx \mathbb{R}^k
\rightarrow\mathbb{R}$, $\mathop{\mathbb{E}}_{\rho
^k}[\varphi(x)] \leq\mathop{\mathbb{E}}_{\gamma^k}[\varphi(x)]$.
\end{lemma}

\begin{pf}
It is easy to check that $\gamma\preceq\rho$. Then, by Kanter's
lemma $\gamma^k \preceq\rho^k$ and the inequality follows from Fact
\ref{fctpeaked}.
\end{pf}

\section{A PTAS for supremum of Gaussian processes}
Our main theorem, Theorem \ref{thmain}, follows immediately from
Lemma \ref
{lmderandjl} and Theorem \ref{thepsnetintro} applied to the
semi-norm $\varphi
\dvtx \mathbb{R}^k \rightarrow\mathbb{R}$, defined by $\varphi(x) = \sup_{i \leq m} |\langle A, ( v_i),x\rangle|$.
\section{Lower bound for oblivious estimators}\label{secappendix}
We now show that Theorem \ref{thepsnet} is optimal: any oblivious linear
estimator for semi-norms as in the theorem must make at least
$(C/\varepsilon)^{k}$ queries for some constant $C > 0$.

Let $S \subseteq\mathbb{R}^k$ be the set of query points of an oblivious
estimator. That is, there exists a function $f\dvtx \mathbb{R}_+^S
\rightarrow\mathbb{R}_+$ such
that for any semi-norm $\varphi\dvtx \mathbb{R}^k \rightarrow\mathbb
{R}_+$, $f((\varphi(x)\dvtx  x \in S)) =
(1\pm\varepsilon) \mathop{\mathbb{E}}_{Y \leftarrow\mathcal{N}^k}[\varphi
(Y)]$. We will assume that $f$
is monotone in the following sense: $f(x_1,\ldots,x_{|S|}) \leq
f(y_1,\ldots,y_{|S|})$ if $0 \leq x_i \leq y_i$ for all $i$. This is
clearly true for any linear estimator (and also for the median
estimator). Without loss of generality suppose that $\varepsilon< 1/4$.

The idea is to define a suitable semi-norm based on $S$: define
$\varphi\dvtx \mathbb{R}
^k \rightarrow\mathbb{R}$ by $\varphi(x) = \sup_{u \in S}|\langle
u /\llVert u\rrVert _2,x\rangle|$. It is
easy to check that for any $v \in S$, $\llVert v\rrVert _2 \leq
\varphi(v)$.
Therefore, the output of the oblivious estimator when querying the
Euclidean norm is at most the output of the estimator when querying
$\varphi$. In particular,
%
%
\begin{eqnarray}\label{eqapp1}
(1-\varepsilon) \mathop{\mathbb{E}}_{Y \leftarrow\mathcal{N}^k}\bigl[\llVert Y\rrVert
_2\bigr] &\leq& f\bigl(\bigl(\llVert x\rrVert _2\dvtx  x \in S
\bigr)\bigr) \leq f\bigl(\bigl(\varphi(x)\dvtx  x \in S\bigr)\bigr)
\nonumber\\[-8pt]\\[-8pt]
&\leq& (1+\varepsilon)\mathop{\mathbb{E}}_{Y
\leftarrow\mathcal{N}^k}\bigl[\varphi(Y)\bigr].\nonumber
\end{eqnarray}
We will argue that the above is possible only if $|S| > (C/\varepsilon
)^k$. Let $\mathcal{S}^{k-1}$ denote the unit sphere in $\mathbb
{R}^k$. For the remaining
argument, we shall view $Y \leftarrow\mathcal{N}^k$ to be drawn as $Y
= R X$, where
$X \in\mathcal{S}^{k-1}$ is uniformly random on the sphere, and $R
\in\mathbb{R}$ is
independent of $X$ and has a Chi-squared distribution with $k$ degrees
of freedom. Let $S(\varepsilon) = \bigcup_{u \in S} \{y \in\mathcal
{S}^{k-1}\dvtx  | \langle u /\llVert u\rrVert _2,y\rangle| \geq1 -
4\varepsilon\}$.

Now,\vspace*{1pt} by a standard volume argument, for any $y \in\mathcal{S}^{k-1}$,
$\pr_X[|\langle X,  y\rangle| \geq1 - 4\varepsilon] <
(O(\varepsilon))^k$. Thus, by a union bound,
$p = \pr_X[X \in S(\varepsilon)] < |S| \cdot(O(\varepsilon))^k$. Further,
for any $y \in\mathcal{S}^{k-1}\setminus S(\varepsilon)$, $\varphi
(y) < 1-4\varepsilon$. Therefore,
\begin{eqnarray*}
\mathop{\mathbb{E}}_{X}\bigl[\varphi(X)\bigr] &=& \pr\bigl[X \notin S(
\varepsilon)\bigr] \cdot \mathbb{E}\bigl[\varphi(X) | X \notin S(\varepsilon)
\bigr]
\\
&&{}  + \pr\bigl[X \in S(\varepsilon)\bigr] \cdot\mathbb {E}\bigl[\varphi(X) | X
\in S(\varepsilon)\bigr]
\\
&\leq&
(1-p) (1-4\varepsilon) + p.
\end{eqnarray*}
Thus
%
%
\begin{eqnarray}\label{eqapp2}
\mathbb{E}\bigl[\varphi(Y)\bigr] &=& \mathbb{E}\bigl[\varphi(R X)
\bigr] = \mathbb{E}[R] \cdot\mathbb{E}\bigl[\varphi(X)\bigr]
\nonumber\\[-8pt]\\[-8pt]
&\leq& \mathbb{E}\bigl[\llVert Y\rrVert _2\bigr] \cdot\bigl((1-p) (1-4\varepsilon) + p\bigr).\nonumber
\end{eqnarray}
Combining equations (\ref{eqapp1}) and (\ref{eqapp2}), we get
\[
1 - \varepsilon\leq(1+\varepsilon)\cdot\bigl((1-p) (1-4\varepsilon) + p\bigr) <
1-3\varepsilon+ 2p.
\]
As $p < |S| \cdot(O(\varepsilon))^k$, the above leads to a contradiction
unless $|S| > (C/\varepsilon)^{k}$ for some constant $C > 0$.




\printaddresses

\end{document}